\documentclass[12pt,leqno]{article}
\usepackage{amsfonts,amssymb,amsthm,amsmath} 
\usepackage[latin1]{inputenc}
\usepackage{hyperref}
\usepackage{dsfont}

\newtheorem{lemma}{Lemma}

\newtheorem{theorem}[lemma]{Theorem}

\begin{document}

\author{Lyonell Boulton$^1$ \\ {\small \texttt{L.Boulton@hw.ac.uk}} \and
Maria Pilar Garcia del Moral$^2$ \\ {\small \texttt{maria.garciadelmoral@uantof.cl}} \and
Alvaro Restuccia$^{2,3}$ \\ {\small \texttt{alvaro.restuccia@uantof.cl}}}

\title{Massless ground state for a compact $SU(2)$ matrix model in 4D}

\date{17/3/15}

\maketitle

\abstract{ 
We show the existence and uniqueness of a massless supersymmetric ground state wavefunction of a $SU(2)$ matrix model in a bounded smooth domain with Dirichlet boundary conditions. This is a gauge system and we provide a new framework to  analyze the quantum spectral properties of this class of supersymmetric matrix models subject to constraints which can be generalized for arbitrary number of colors.}

\footnotetext[1]{Maxwell Institute for Mathematical Sciences and Department of Mathematics, Heriot-Watt University, Edinburgh, EH14 4AS, United Kingdom.}
\footnotetext[2]{Departamento de F\'isica, Universidad de Antofagasta, Antofagasta, Aptdo 02800, Chile.}
\footnotetext[3]{Departamento de F\'isica, Universidad Sim\'on Bol\'ivar, Valle de Sartenejas, 1080-A Caracas, Venezuela.}

\newpage

\section{Introduction} Determining the ground state of the $(0+1)$ matrix models is a longstanding open problem \cite{halpern, hoppe, dwhn}. This is a subject of interest in different areas, including: matrix models \cite{bfss}, Yang-Mills theories \cite{yi, sethi-stern, porrati} and M-theory \cite{hoppe, fh, bgmr}. In this paper we show the existence of the ground state wavefunction of a matrix
model which takes values on a compact space and is subject to the constraints associated to gauge symmetry. This model will serve as a benchmark in order to illustrate a new method to demonstrate  the existence and uniqueness of massless supersymmetric ground state wavefunctions of supersymmetric matrix models at finite temperature. 

Boundary conditions can change dramatically the spectral properties of field theories and at the same time they may imply that certain symmetries may be partially or even totally broken. Indeed,  for supersymmetric theories ($N=2$) only periodic, Dirichlet and Neumann boundary conditions preserve partial supersymmetry $N=1$, \cite{manolo}. These type of models can be of interest to test certain aspects of AdS/CFT dualities, such as  characterizing black holes at finite temperature.  Yang-Mills or QED fields on a box have been largely considered, for example in  bag-models or compact QED models \cite{abls}, to study confinement properties of hadrons, phase transitions and chiral symmetry breaking. In a similar fashion as QCD, compact QED in a strong-coupling region exhibits charge confinement as well as spontaneous chiral symmetry breaking. Both these aspects are thought to be related to the existence of monopoles. Many authors have studied this problem in the past, see for example \cite{Takahashi} and the references there in. The nature of the phase transitions in compact QED has in fact been under debate for a long time.

The deconfinement temperature of finite-temperature compact electrodynamics in $2+1$ dimensions has been shown to be insensitive to external fields \cite{ics}. Matrix models of the compact QED are likely to be relevant in the analysis of these theories, subject to a slow-mode regime. Effectively, they serve as toy models of such a phenomenon. 
From the M-theory point of view, $SU(2)$ matrix models correspond to a supermembrane (regularized via $SU(2)$) propagating
in a 4-dimensional Minkowski space-time \cite{dwhn, dwln}. The $SU(2)$ ground state wave function has been examined also in \cite{hlt2,su2, hlt}. The question of whether a unique ground state with zero energy exits for a Minkowski spacetime was not completely settled in \cite{dwhn}. 

In \cite{octonions} we addressed the existence and uniqueness of the ground state wavefunction for unconstrained models restricted to Dirichlet boundary conditions. Determining the ground state wave function from the point of view of the supermembrane is an open question which was posed when the model was originally formulated. This is expected to corresponds to an 11D wavefunction constructed in terms of the 11D supermultiplet of supergravity. The hamiltonian of the supermembrane has two independent contributions one associated to the movement of the center of mass in 11D Minkowski the space-time and a second one associated to the supermembrane excitations. The 
 existence and uniqueness of the ground state wave function of the 11D supermembrane in the case when the massive excitations are forced
to lie in a compact space with Dirichlet boundary conditions will be
analyzed elsewhere 
\cite{bgmr6}. 

In this paper we examine a particular model, the $SU(2)$ gauge supermembrane. Our goal is to draft the relevant step towards finding the ground state wave function of the regularized supermembrane on a 11D flat background\footnote{The mass operator of the regularized supermembrane in a 11D Minkowski spacetime corresponds to the $N=16$ supersymmetric matrix model.}, assuming that the center of mass of the supermembrane propagates freely in the 4D spacetime but its membrane excitations are confined to a compact space of arbitrary large radius $R$. We will show that, given a Dirichlet boundary condition, there exists a unique massless ground state wave function for the mass operator of the model.

The paper is organized as follows. In Section~2 we present the $SU(2)$ matrix model, the mass operator, the constraint associated to the local $SU(2)$ symmetry and the supersymmetric algebra. In Section~3 we formulate the Dirichlet problem. In Section~4 we prove the existence and uniqueness of the ground state. A final section is devoted to highlighting our main conclusions.


\section{The $SU(2)$ regularized matrix model} 
The model we will consider was introduced in \cite{dwhn} and it arises from the 11D mass operator by taking  all the fields to zero  $X^{mA}=\lambda_{\alpha'}^A=0$ for $m=1,\dots,7$ and $\alpha'=1,\dots,7$, except those transforming under $U(1)$ via $(Z^A,\overline{Z}^A, \lambda^A_8=\lambda^A)$. In conjunction with the light cone gauge fields $X^+$ and $X^-$, those fields describe a regularized supermembrane propagating in a 4D space-time. We will also restrict to the simplest nonabelian gauge symmetry, given by $G=SU(2)$.
The ground state wave function expressed as a superfield admits and expansion in the superfields that does not admit a factorization. The solution may have an even or odd number of odd grassmanian coordinates. The even ones have the expression
\begin{equation}\label{wf3}
\Psi=\phi_0(Z,\overline{Z})+\epsilon^{ABC}\phi^A(Z,\overline{Z})\lambda^B\lambda^C.
\end{equation}
This is constructed in terms of four functions $(\phi_0,\phi_A)$ with $A=1,2,3$.

The associated hamiltonian is 
\[H=-\nabla^2+V_B+V_F= -\frac{\partial^2}{\partial Z_A\partial\overline{Z}^A}+ V(Z, \overline{Z},\lambda)\]
with 
\[\begin{aligned}V_B &=\frac{1}{4}\epsilon^E_{AB}\epsilon_{CDE}\left(2Z^A\overline{Z}^B\overline{Z}^CZ^D\right) \qquad \text{and}\\
V_F&=\frac{1}{\sqrt{2}}\epsilon_{ABC}\left(Z^A\lambda^B\lambda^C-\overline{Z}^A\frac{\partial}{\partial\lambda_B}\frac{\partial}{\partial\lambda_C}\right),\end{aligned}\]
subject to the first class constrain $\varphi^A\vert \Psi \rangle=0$ where
\[\varphi^A=\epsilon^{ABC}\left(Z^B\frac{\partial}{\partial Z^C}+\overline{Z}^B\frac{\partial}{\partial \overline{Z}^C}+\lambda_B\frac{\partial}{\partial\lambda_C}\right).\]
 The associated supercharges are the following
\begin{equation}
\begin{aligned}
Q&=\sqrt{2}\frac{\partial}{\partial Z^A}\frac{\partial}{\partial\lambda_A}-\epsilon_{ABC}Z^A\overline{Z}^B\lambda^C\\
Q^{\dag}&=-\sqrt{2}\frac{\partial}{\partial \overline{Z}^A}\lambda^A+\epsilon_{ABC}Z^A\overline{Z}^B\frac{\partial}{\partial\lambda^C}
\end{aligned}
\end{equation} 
and the superalgebra satisfies the conditions
\[
\{Q,Q\}=2\sqrt{2}\overline{Z}^A\varphi_A,\quad
\{Q^{\dagger},Q^{\dagger}\}=2\sqrt{2}Z^A\varphi_A \quad \text{and} \quad
\{Q,Q^{\dagger}\}=2H.
\]

\section{Existence and uniqueness of a solution in a compact domain}

We now consider the existence and uniqueness of the ground state wave function for a system that is restricted by a first class constraint given by the $SU(2)$ gauge condition. We proceed in a similar fashion as in \cite{octonions}, but the presence of a constraint implies an added difficulty to the wavefunction analysis.  Let us firstly show that there is no need to solve the constraint in this type of matrix models explicitly. This simplifies the arguments significantly.  We use the representation of the wave function in terms of an anticommuting grassman coordinate (Lemma~\ref{lema1}) or its representation in the Fock space (Lemma~\ref{lema2} and Theorem~\ref{teorema1}) when convenient.

The constraint $\varphi^A\Psi=0$ defines a closed subspace of the Sobolev Hilbert space $\mathcal{H}^1(\Omega)$. Denote by $X$ the closure of this subspace in the norm of $L_2(\Omega)$. For the $SU(2)$ regularized supermembrane in four dimensions (RSM) we are interested in the following homogeneous problem. Given $g\in \mathcal{H}^2(\overline{\Omega})\cap X$, find $\Phi\in \mathcal{H}^2(\Omega)$ such that 
\begin{equation} \label{I}
\begin{cases}
(-\nabla^2+ V)\Phi=0 & \textrm{in } \Omega\\
\Phi=g\quad & \textrm{on } \partial\Omega\\
\Phi\in X.
\end{cases}
\end{equation}
Here $V$ is the potential of the hamiltonian $H$ of the RSM. 
We call $\Phi$ the ground state wavefunction of the hamiltonian in $\Omega$ since it corresponds to the restriction to $\Omega$ of the ground state wavefunction of the hamiltonian in $R^{D(N^2-1)}$, with $D=2$ and $N=2$. Besides, $\Phi$ minimizes the Dirichlet form, associated to the hamiltonian of RSM among the states which satisfy the constraint and the boundary condition.
\newline
Let $D(\Lambda,\Phi)$, with $\Lambda\in \mathcal{H}^1(\Omega)$ and $\Phi\in \mathcal{H}^1(\Omega)$, be the Dirichlet form associated to the operator $-\nabla^2+V$. It is defined by
\begin{equation}
D(\Lambda,\Phi)=(\nabla \Lambda,\nabla\Phi)+ (\Lambda,V\Phi)
\end{equation}
where $(\cdot,\cdot)$ denotes the internal product in $L_2(\Omega)$.
In particular if $\chi\in C_0^{\infty}(\Omega)$ we have
\begin{equation}\label{chi}
D(\chi,\chi)=(\chi,(-\nabla^2+V)\chi)\ge 0
\end{equation}
due to the supersymmetric structure of the mass operator. It then follows
\begin{equation}\label{varphi} D(\varphi,\varphi)\ge 0\end{equation}
for all $\varphi\in \mathcal{H}_0^1(\Omega).$\newline

Let $\Phi\in \mathcal{H}^1(\Omega)\cap X$, $\phi=g$ on $\partial\Omega$ be a solution of the Dirichlet problem 
\begin{equation}\label{v-p}D(\varphi,\Phi)=0\end{equation}
for all $\varphi\in \mathcal{H}_0^1(\Omega).$ 
 Then for any $\Lambda\in \mathcal{H}^1(\Omega)\cap X$, $\Lambda=g$ on $\partial\Omega$ we obtain
$$\Lambda-\Phi=\varphi,$$  where $\varphi\in \mathcal{H}_0^1(\Omega)\cap X.$
It follows that 
$$D(\Lambda,\Lambda)=D(\Phi,\Phi)+D(\Phi,\varphi)+D(\varphi,\Phi)+D(\varphi,\varphi),$$
we now use (\ref{varphi}) and (\ref{v-p}) to get
\begin{equation}
D(\Lambda,\Lambda)\ge D(\Phi,\Phi).
\end{equation}
Consequently $D(\Phi,\Phi)$ is the minimum of the values of the Dirichlet form evaluated on the states $\Lambda\in \mathcal{H}^1(\Omega)\cap X$, $\Lambda=g$ on $\partial\Omega$. This is analogous to the Dirichlet principle in  Electrostatics. We are going to show that there exists a unique $\Phi$ solution to the Dirichlet problem (\ref{v-p}), moreover the solution $\Phi\in \mathcal{H}^2(\Omega)\cap X$.  We may then integrate by parts in (\ref{v-p}) to obtain a unique solution to (\ref{I}). The minimum of the Dirichlet form is then obtained at the solution of problem (\ref{I}).\newline 
Let $f:=(\nabla^2-V)g$.
The following inhomogeneous problem is a re-formulation of \eqref{I}.  Find 
$\Psi\in \mathcal{H}_0^1(\Omega)\cap \mathcal{H}^2(\Omega)$
 such that 
\begin{equation} \label{II}
\begin{cases}
(-\nabla^2+ V)\Psi=f & \textrm{in } \Omega\\
\Psi=0 & \textrm{on } \partial\Omega\\
\Psi\in X.
\end{cases}
\end{equation}
If $\Psi$ is a solution of \eqref{II}, then $\Phi=\Psi+g$ is a solution of \eqref{I}. \newline

In the following we will take $\Omega$ to be a ball of radius $R\ne0$. If 
$\Psi\in \mathcal{H}_0^1(\Omega)\cap \mathcal{H}^2(\Omega)$ then $\varphi^A\Psi\in \mathcal{H}_0^1(\Omega)$. Geometrically it means that $\Omega$ remains invariant under the symmetry generated by the first class constraint of the theory.
 The main result of this work is given by the following theorem:
\begin{theorem} \label{teorema1} 
Let $g\in \mathcal{H}^2(\Omega)\cap X$. There exists a unique solution for the problem~ \eqref{II} which lies in $\mathcal{H}_0^1(\Omega)\cap \mathcal{H}^2(\Omega)\cap X$. Consequently, there exists a unique solution $\Phi\in \mathcal{H}^2(\Omega)\cap X$ to the problem~\eqref{I}.
\end{theorem}

The proof of this theorem relies on two auxiliary lemmas.

\begin{lemma} \label{lema1}
Let $Q$ and $Q^\dag$ be the supercharge operators associated to the RSM.
If $\Psi\in \mathcal{H}_0^1(\Omega)$ satisfies the conditions 
\begin{equation}    \label{kernel}
Q\Psi=Q^{\dag}\Psi=0 \qquad \text{in} \qquad \Omega,
\end{equation}
then $\Psi=0$ in $\overline{\Omega}$. \label{eql}\end{lemma}

\begin{proof}
Assume \eqref{kernel}.
According to the conditions of the supersymmetric algebra, we have 
\begin{equation}
H\Psi=0 \qquad \textrm{in}\quad \Omega.
\end{equation}
By elliptic regularity, it immediately follows that $\Psi\in \mathcal{H}_0^1(\Omega)\cap \mathcal{H}^2(\Omega)$. Moreover, \eqref{kernel}, which originally held true in $\Omega$, can be extended smoothly to the boundary, $\partial\Omega$. 

Define $\rho^2=Z^A\overline{Z}^A$ and $\Psi_{\rho^2}= \frac{\partial\Psi}{\partial\rho^2}$, the normal derivative at 
\[
\partial\Omega=\{Z^A, \overline{Z}^A: Z^A\overline{Z}^A= R^2\}.
\] 
Then using (\ref{kernel}) on $\partial\Omega$ we obtain
$$\frac{\partial}{\partial Z^A}\frac{\partial}{\partial \lambda_8^A}\Psi=0\quad \text{and} \quad\frac{\partial}{\partial\overline{Z}^A}\lambda_8^A\Psi=0\quad\textrm{on}\quad\partial\Omega.$$
Rewriting the latter in terms of $\Psi_{\rho^2}$, and using that all the tangential derivatives at $\partial\Omega$ are zero, gives
\begin{equation*}
\overline{Z}^A\frac{\partial}{\partial\lambda_8^A}\Psi_{\rho^2}=0 \quad \text{and} \quad Z^A\lambda_8^A\Psi_{\rho^2}=0\quad\textrm{on\quad}\partial\Omega.
\end{equation*}
From this we conclude that $Z^A\overline{Z}^A\Psi_{\rho^2}=0$ on $\partial\Omega$. That is $R^2\Psi_{\rho^2}=0$ for any $R^2\ne 0.$ Consequently $\Psi_{\rho^2}=0$ on $\partial\Omega$. 

We therefore have $\Psi=0$ and $\Psi_{\rho^2}=0$ on $\partial\Omega$. By virtue of the Cauchi-Kovalewski Theorem, it follows that $\Psi=0$ in a neighborhood of $\partial\Omega$. Since $V$ is analytic in $\Omega$,  in fact $\Psi=0$ in $\overline{\Omega}$. 
\end{proof}
\begin{lemma} \label{lema2}
Let $f\in L_2(\Omega)\cap X=X$, there always exists a solution $\Psi\in \mathcal{H}_0^1(\Omega)$ to the Dirichlet problem
\begin{equation}\label{12}
\mathcal{D}(\chi,\Psi)=(\chi,f) \quad\textrm{for all}\quad \chi\in \mathcal{H}_0^1(\Omega),\end{equation}
Remark: The regularity properties of the Dirichlet form ensure that $\Psi\in \mathcal{H}_0^1(\Omega)\cap\mathcal{H}^2(\Omega).$
\end{lemma} 

\begin{proof}
For the hamiltonian of the RSM, $\mathcal{D}$ is coercive since $V$ is bounded from below in $\Omega$, see \cite{Folland}. We may then use the theorem (7.21) from \cite{Folland}, which states that $\Psi$ exists, provide the subspace \[K=\{\xi\in \mathcal{H}_0^1(\Omega): \mathcal{D}(\chi,\xi)=0\quad \textrm{for all}\quad \chi\in \mathcal{H}_0^1(\Omega)\}\] is orthogonal in $L_2(\Omega)$ to $f$.

From the regularity property of the Dirichlet form, we obtain $\xi\in \mathcal{H}_0^1(\Omega)\cap C^{\infty}(\Omega)$, hence $(-\nabla^2 +V)\xi=0$ in $\Omega$. Consequently $Q\xi=0$ and $Q^{\dag}\xi=0$ in $\Omega$. According to Lemma~\ref{lema1}, we conclude that $\xi=0$ in $\Omega$, hence $K=\{0\}$ is orthogonal to $f$. The regularity properties of elliptic operators ensure that $\Psi\in \mathcal{H}_0^1(\Omega)\cap \mathcal{H}^2(\Omega).$
\end{proof}

The proof of Theorem~\ref{teorema1} is completed as follows.
According to Lemma~\ref{lema2}, there exists a solution $\Psi\in \mathcal{H}_0^1(\Omega)\cap \mathcal{H}^2(\Omega)$ to the Dirichlet problem~\eqref{12}. By integration by parts 
\begin{equation}
(-\nabla^2+V)\Psi=(\nabla^2-V)g \quad\textrm{in}\quad \Omega.
\end{equation}
Given $\chi\in C_0^{\infty}(\Omega)$ we consider 
\begin{equation}
(\varphi^A \chi, (-\nabla^2+V)\Psi)=(\varphi^A\chi, (\nabla^2-V)g)
\end{equation}
 then, using that $(-\nabla^2+V)$ commutes with $\varphi^A$, we get 
\begin{equation}
( (-\nabla^2+V)\chi,\varphi^A\Psi)=((\nabla^2-V)\chi,\varphi^A g)=0.
\end{equation}
 We thus have
\begin{equation}
D(\chi,\varphi^A \Psi)=0
\end{equation}
for all $\chi\in C^{\infty}_0(\Omega)$ and hence, taking limits, for all $\chi\in \mathcal{H}^1_0(\Omega).$ Consequently, using that $\varphi^A\Psi\in \mathcal{H}_0^1(\Omega)$, we obtain $\varphi^A\Psi=0,$ that is 
$\Psi\in\mathcal{H}_0^{1}(\Omega)\cap\mathcal{H}^2(\Omega)\cap X$ is a solution to the problem~\eqref{I}. It is unique, since 
$K\subset\{0\}$.

\section{Discussion}
Let us compare our analysis, on a smooth bounded domain, with the analysis presented in \cite{su2} and \cite{queiroz} on unbounded regions. In \cite{su2} the ground state of the Yang-Mills quantum mechanics is considered, and upper and lower bounds for the minimum eigenvalue are obtained. The hamiltonian has a similar structure to the hamiltonian of the bosonic membrane, its potential has a quartic dependence on the configuration variables with valleys extending to infinity. The hamiltonian is bounded from below by a hamiltonian with a basin shaped potential which has a discrete spectrum with non-zero minimal eigenvalue \cite{simons, luscher}. 

Another approach for the same problem was considered in \cite{inertia} using the Molchanov mean value condition for the potential \cite{molchanov, mazya}. 
The discreteness of the spectrum is originated by the non-zero ground-state energy of the bosonic harmonic oscillator. This argument cannot be extended to supersymmetric matrix models because the supersymmetric harmonic oscillator has zero ground state energy.  In the supersymmetric matrix models describing the regularized supermembrane the spectrum is continuous from $[0,\infty)$ and the main problem is to analyze whether $0$ is an eigenvalue or not.

The boundary condition in \cite{su2} prescribes that the wave functions should decay at infinity. We believe that this may be too restrictive for the quantum mechanics of $SU(2)$ supermembrane in an unbounded domain. As emphasizes in \cite{dwhn} one has to consider the whole Hilbert space $L_2(R^D)$ which includes wave functions which are square integrable functions but diverge to infinity in some directions. This problem is still open, and it requires a more delicate analysis of the boundary conditions at infinity. For that reason the problem formulated on a bounded domain with a nontrivial boundary condition is of relevance. 

The authors of \cite{queiroz} constructed a matrix model of QCD incorporating non-trivial topological aspects of the theory. The hamiltonian matrix model has similar properties to the bosonic hamiltonian of the supermembrane with non-trivial central charges or non-trivial winding \cite{bgmr0, bgmr}. In the latter paper, the fermionic potential is a relatively bounded operator with respect to the bosonic hamiltonian, hence the spectrum of the supersymmetric matrix model has the same qualitative behaviour as its bosonic sector \cite{bgmr}. However, the main problem in this paper is to resolve the spectrum of the $D=11$ supermembrane with zero winding in which case one expects that the ground state should correspond to the $D=11$ supergravity multiplet.

The approach we have in mind for the extension of our analysis to the unbounded region $R^D$ consists in three steps. The first one is to solve  the 'internal' problem that is the existence and uniqueness of the ground state on a smooth bounded domain $\Omega$. The second step is to analyze the problem on the exterior region of $\Omega$. Finally we will consider the matching of the two solutions. So far we have solved the first step for the $SU(2)$ problem and expect to generalize these arguments for the $SU(N)$ regularization of the $D=11$ Supermembrane. The analysis of the exterior problem introduces a new aspect compared to the first step. It is the behaviour of the potential at infinity. The matching conditions, a well-developed topic in elliptic partial differential equations, will give the final answer concerning to the existence or not of the ground state wave function of the $D=11$ Supermembrane.  A different approach is to consider a sequence of balls of  increasing radius  and to show the convergence of the sequence of solutions under a suitable  a priori boundary condition function $g\in L_2(R^D)$.

\section{Conclusions}
We presented an $SU(2)$ gauge supersymmetric matrix model whose center of mass propagates freely in a 4D spacetime with its transversal oscillations restricted to a compact space subject to Dirichlet boundary conditions. We establish the existence and uniqueness of the massless ground state of the theory. Our analysis relies heavily on the property that the constraint is of first class, which is the standard case of gauge theories. Extension of this analysis to matrix models subject to second class constraints are worth of further study, and are beyond of the scope of the present paper. 

Our proofs simplify significantly, due to the fact that we do not require to solve explicitly the constraint. The approach we have illustrated here can be extended in order to determine the ground state wavefunction of other matrix models associated to supersymmetric gauge systems, such as AdS/CFT at finite temperatures or QED matrix models in compact spacetime.

\section{Aknowledgements}  MPGM would like to thank to Manuel Asorey for interesting comments and to the Theoretical Physics Department at U. Zaragoza, Spain, for kind invitation while part of this work was done. MPGM is supported by Mecesup ANT1398, Universidad de Antofagasta, (Chile). A.R. is partially supported by Projects Fondecyt
1121103 (Chile).

\end{document}